\documentclass{article}
\usepackage[utf8]{inputenc}
\usepackage[dvipdfmx]{graphicx}
\usepackage{fullpage}
\usepackage{amsthm}
\usepackage{amsmath}
\usepackage{amssymb}
\usepackage{mathtools}
\usepackage[colorlinks,citecolor=blue,bookmarks=true,linktocpage]{hyperref}
\usepackage{aliascnt}
\usepackage[numbered]{bookmark}
\usepackage[capitalise]{cleveref}
\usepackage[backend=biber,url=false,arxiv=abs,maxbibnames=100,isbn=false,sortcites=true]{biblatex}
\usepackage{cleveref}
\usepackage{xcolor}
\usepackage{physics}

\newtheorem{theorem}{Theorem}
\newtheorem{lemma}{Lemma}
\newtheorem{corollary}{Corollary}

\newtheorem{observation}{Observation}

\newcommand{\prb}[1]{\textnormal{\scshape #1}}

\title{On the complexity of finding a spanning even tree in a graph\thanks{This work is partially supported by JSPS KAKENHI Grant Numbers
JP20K04973, 
JP20H00595, 
JP20H05964, 
JP20H05967, 
JP21K19765, 
JP21K17707, 
JP22H00513, 
JP23K28034, 
JP23H04388, 
JP24H00686, 
JP24H00697, 
JP24K02898, 
and by JST, CRONOS, Japan Grant Number JPMJCS24K2. 
}
}

\author{
Tesshu Hanaka\thanks{Kyushu University. Email: \texttt{hanaka@inf.kyushu-u.ac.jp}} \and
Yasuaki Kobayashi\thanks{Hokkaido University. Email: \texttt{\{koba, seto\}@ist.hokudai.ac.jp}} \and
Kazuhiro Kurita\thanks{Nagoya University. Email: \texttt{kurita@i.nagoya.ac.jp}, \texttt{ono@nagoya-u.jp}} \and
Yasuko Matsui\thanks{Tokai University. Email: \texttt{yasuko@tokai.ac.jp}} \and
Atsuki Nagao\thanks{Ochanomizu University. Email: \texttt{a-nagao@is.ocha.ac.jp}} \and
Hirotaka Ono\footnotemark[3] \and
Kazuhisa Seto\footnotemark[2]
}

\begin{document}

\maketitle

\begin{abstract}
A tree is said to be \emph{even} if for every pair of distinct leaves, the length of the unique path between them is even.
In this paper we discuss the problem of determining whether an input graph has a spanning even tree.
Hofmann and Walsh [Australas. J Comb. 35, 2006] proved that this problem can be solved in polynomial time on bipartite graphs.
In contrast to this, we show that this problem is NP-complete even on planar graphs.
We also give polynomial-time algorithms for several restricted classes of graphs, such as split graphs, cographs, cobipartite graphs, unit interval graphs, and block graphs.
\end{abstract}

\section{Introduction}

A \emph{spanning tree} is a fundamental combinatorial object in graphs.
A natural optimization problem regarding this concept is \prb{Minimum Spanning Tree}, where the goal is to find a minimum weight spanning tree of a given edge-weighted graph, which is well known to be solvable in polynomial time.
Given this fact, there are numerous studies aiming at finding a (minimum) spanning tree with some additional constraints and/or objectives~\cite{PapadimitriouY82,GareyJ79,BercziKKYY24,HassinL04,Goemans06,HanakaK23}.

In this paper, we focus on yet another constraint: a \emph{parity constraint}.
A tree is said to be \emph{even} (resp.\ \emph{odd}) if for every pair of distinct leaves, the distance between them is even (resp.\ \emph{odd}), where the distance is defined as the number of edges in the unique path between them.
A \emph{spanning even tree} in a graph $G$ is a spanning tree that is even.
A \emph{spanning odd tree} is defined analogously.
It is not hard to see that every spanning odd tree in $G$ is a Hamiltonian path: If a spanning tree has at least three leaves, one of the pairs must be connected by a path of even length.
Hence, the problem of determining whether a graph has a spanning odd tree is equivalent to \prb{Hamiltonian Path}, which is known to be NP-complete even on bipartite graphs~\cite{ItaiPS82}.
However, the complexity of finding a spanning even tree is not much known in the literature.
This problem was studied by Hoffman and Walsh~\cite{EvenST:HoffmanW06}, who gave a polynomial-time algorithm that decides whether a given bipartite graph has a spanning even tree using a polynomial-time algorithm for \prb{Matroid Intersection}.\footnote{A faster solution is given in the commentary on Problem H in the 2024 ICPC Asia Yokohama Regional Contest. See~\url{https://icpc.iisf.or.jp/past-icpc/regional2024/commentaries-2024.pdf}.}
To the best of our knowledge, there are no further algorithmic or complexity-theoretic results on this problem.
In particular, the complexity for general graphs was open.

Several combinatorial problems with parity requirements have interesting distinctions from their ``conventional'' ones.
The even path problem, which asks to find a path of even length between two specified vertices, is a famous example of this kind.
This problem is ``equivalent'' to the odd counterpart, as there is a trivial reduction from one to the other.
The even path problem can be reduced to a perfect matching problem~\cite{LapaughP84}, which can be solved in polynomial time.
The (shortest) even/odd cycle problems in directed graphs are other notable examples: For the odd case, it is easy to find a (shortest) odd cycle in a directed graph by finding a shortest odd walk, which can be found in polynomial time, while the even counterpart is considerably less trivial~\cite{BjorklundHK22}.

Apart from the computational perspective, Jackson and Yoshimoto~\cite{JacksonY24} discussed several conditions under which regular graphs have a spanning even tree.
They showed that every regular bipartite graph has no spanning even trees at all.
They also conjectured that every regular non-bipartite connected graph has a spanning even tree and verify this conjecture for some spacial case. 
This conjecture was eventually confirmed affirmatively by~\cite{ellingham2024spanningweaklytreesgraphs}.

Motivated by the work of Hoffman and Walsh~\cite{EvenST:HoffmanW06} and Jackson and Yoshimoto~\cite{JacksonY24}, we address the complexity of the problem of finding a spanning even tree in a graph.
We show that the problem of deciding whether a given graph has a spanning even tree is NP-complete even on planar graphs.
To circumvent this intractability, we consider the problem on several graph classes.
For cographs, cobipartite graphs, unit interval graphs, and split graphs, we give simple necessary and sufficient conditions for having spanning even trees.
These conditions immediately enable us to find a spanning even tree in polynomial time on these classes of graphs.
For block graphs, we design a slightly nontrivial algorithm to solve the problem by exploiting their block-cut tree structure.

\section{Preliminaries}

Throughout this paper, we assume that all graphs are simple.
Let $G = (V, E)$ be a graph.
For $v \in V$, we denote by $N_G(v)$ the set of (open) neighbors of $v$ in $G$, that is, $N_G(v) = \{u \in V : \{u, v\} \in E\}$.
This notation is extended to sets: $N_G(X) = \bigcup_{v \in X}N_G(v) \setminus X$ for $X \subseteq V$.

A \emph{tree} is a connected graph without any cycles.
A \emph{leaf} of a tree (or more generally, of a graph) if a vertex of degree~$1$. 
A tree is said to be \emph{even} if for every pair of distinct leaves. the number of edges of the unique path between them is even.
Every even tree can be characterized in the following form.
\begin{observation}\label{obs:bichromatic-spanning-tree}
    A tree $T$ is even if and only if it has a proper $2$-coloring with black and white such that all the leaves are colored in black.
\end{observation}
We refer to such a coloring as an \emph{admissible coloring for $T$}. 
Note that for any even tree $T$, the admissible coloring for $T$ is uniquely determined.

\paragraph{Graph classes.} 
A \emph{split graph} is a graph $G = (V, E)$ such that $V$ can be partitioned into a clique $K$ and an independent set $I$.
Given a split graph, one can compute such a partition in linear time~\cite{HammerS81}.

A \emph{cobipartite graph} is the complement of a bipartite graph.
It is easy to see that for every cobipartite graph with at least two vertices, its vertex set can be partitioned into two cliques.

A \emph{cograph} is inductively defined as follows:
\begin{itemize}\setlength{\parskip}{0cm}
    \item a single-vertex graph is a cograph,
    \item for two vertex-disjoint cographs $G_1 = (V_1, E_1)$ and $G_2 = (V_2, E_2)$, the disjoint union $G_1 \cup G_2 = (V_1 \cup V_2, E_1 \cup E_2)$ of them is a cograph, and
    \item for two vertex-disjoint cographs $G_1 = (V_1, E_1)$ and $G_2 = (V_2, E_2)$, the complete join $G_1 \Join G_2 =  (V_1 \cup V_2, E_1 \cup E_2 \cup E_{\Join})$ of them is a cograph, where $E_{\Join} = \{\{u, v\} : u \in V_1, v \in V_2\}$.
\end{itemize}

Given a connected cograph $G$ with at least two vertices, one can compute two cographs $G_1$ and $G_2$ such that $G = G_1 \Join G_2$ in linear time~\cite{CorneilPS85}\footnote{More precisely, we can compute a rooted tree representation, called a \emph{cotree}, of a cograph in linear time~\cite{CorneilPS85}.}.

A graph $G = (V, E)$ is said to be \emph{interval} if it has an interval representation: An interval representation of $G$ consists of a set of intervals $\mathcal I = \{I_v : v \in V\}$ on the real line, and two vertices $u$ and $v$ are adjacent to each other in $G$ if and only if the corresponding intervals $I_u$ and $I_v$ intersect.
An interval graph is called a \emph{unit interval graph} if it has a unit interval representation, which is an interval representation that has only intervals of unit length.

A \emph{block graph} is a graph in which every biconnected component induces a clique.
We denote by $\mathcal B$ the set of biconnected components (i.e., maximal cliques) and by $C$ the set of cut vertices in $G$.
A \emph{block-cut tree} of $G$ is a tree $\mathcal T$ with vertex set $\mathcal B \cup C$ such that two vertices $B \in \mathcal B$ and $c \in C$ are adjacent if and only if $c \in B$.
Given a connected block graph, we can compute the block-cut tree of $G$ in linear time~\cite{HopcroftT73}.

\section{NP-hardness of finding a spanning even tree}
In this section, we prove that the problem of finding a spanning even tree is hard in general.

\begin{theorem}\label{thm:hardness}
    The problem of deciding whether an input graph $G$ has a spanning even tree is NP-complete.
\end{theorem}

Obviously, the problem belongs to NP.
In the following, we give a polynomial-time reduction from \prb{Satisfiability}.
Let $\varphi$ be a CNF formula with variable set $X = \{x_1, x_2, \ldots, x_n\}$ and clause set $C = \{c_1, c_2, \ldots, c_m\}$.
We construct a graph $G_\varphi$ such that $G_\varphi$ has a spanning even tree if and only if $\varphi$ is satisfiable.

Before constructing the graph $G_\varphi$, we demonstrate an ``incomplete'' reduction to give an intuition of gadgets used in our proof.
We assume that each clause has at most one of $x_i$ and $\neg x_i$ for each $i$.
For each variable $x_i \in X$, the variable gadget $G_i$ consists of six vertices as in \Cref{fig:variable-gadget}.
The vertex $v_i$ in \Cref{fig:variable-gadget} is called the \emph{connection vertex} of $G_i$.
There are two possible ways to take a spanning even tree in $G_i$ as depicted in \Cref{fig:variable-gadget}.
These spanning even trees represent possible truth assignments $\alpha$ of $x_i$, where the left spanning tree $T_i$ indicates $\alpha(x_i) = \texttt{true}$ and the right spanning tree $F_i$ indicates $\alpha(x_i) = \texttt{false}$.
\begin{figure}
    \centering
    \includegraphics[width=0.4\textwidth]{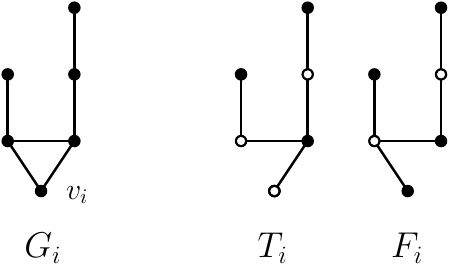}
    \caption{The left figure depicts the variable gadget $G_i$.
    There are two ways to take a spanning even tree in $G_i$.}
    \label{fig:variable-gadget}
\end{figure}
For each clause $c_j \in C$, we use a path $P_j$ of length $2$.
We let $w_j$ be one of the end vertex of the path.
We add a path of length $1$ (i.e., an edge) between $v_i$ and $w_j$ if $c_j$ contains $x_i$ (as a positive literal) and add a path of length $2$ between them if $c_j$ contains $\neg x_i$, which is denoted by $P_{i,j}$.
The graph constructed at this point is illustrated in \Cref{fig:partial}.
\begin{figure}[ht]
    \centering
    \includegraphics{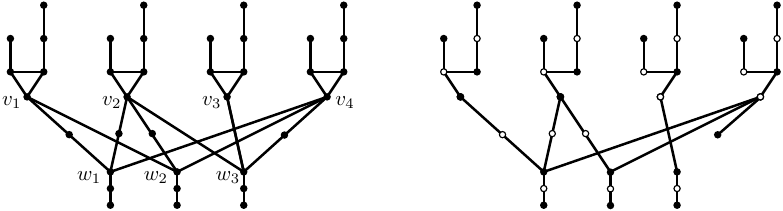}
    \caption{The figures illustrate the (partial) graph for $\varphi = (\neg x_1 \lor \neg x_2 \lor x_4) \land (x_1 \lor \neg x_2 \lor x_4) \land (x_2 \lor x_3 \lor \neg x_4)$ and the subgraph constructed from the assignment $(x_1, x_2, x_3, x_4) = (\texttt{false}, \texttt{false}, \texttt{true}, \texttt{true})$.}
    \label{fig:partial}
\end{figure}

From a satisfying assignment $\alpha$ of $\varphi$, we construct a subgraph of $G_\varphi$ as follows.
For each variable $x_i$, we include all edges in $T_i$ if $\alpha(x_i) = \texttt{true}$ and all edges in $F_i$ otherwise.
We also include all edges in the paths corresponding to the clauses.
Let $x_i$ be a variable and let $c_j$ be a clause that contains a literal $\ell \in \{x_i, \neg x_i\}$.
Suppose that $\ell = x_i$.
If $\alpha(x_i) =\texttt{true}$, we include the edge $\{v_i, w_j\}$.
Suppose otherwise that $\ell = \neg x_i$.
Then we include the edge of $P_{i,j}$ incident to $v_i$ in either case ($\alpha(x_i) = \texttt{true}$ or $\texttt{false}$), and include moreover the other edge of $P_{i,j}$ if $\alpha(x_i)=\texttt{false}$.
The construction from a specific assignment is depicted in \Cref{fig:partial}.
In the following, we use colors black and white to indicate that the parity of the distance from any leaves in the subgraph, while leaves are always colored in black.
It is not hard to observe that there is no odd path between any pair of leaves but the subgraph constructed in the above way may have cycles or more than one components.
In order to obtain a spanning even tree from a satisfying assignment, we use other two types of gadgets: \emph{garbage collector gadgets} and \emph{connector gadgets}.

To break a cycle in the subgraph constructed as above, it suffices to remove an edge from some $P_{i,j}$, as other edges are not included in any cycles.
If a cycle contains an edge of the form $\{v_i, w_j\}$, we can simply remove it while maintaining the property that $w_j$ is connected to some variable gadget.
Moreover, this does not make a new leaf, which maintains the property that there is no odd path between any pair of leaves.
However, we may not simply remove an edge of the path $P_{i,j}$ when it consists of two edges, since the removal may make a new leaf that belongs to an odd path between some leaves.
To avoid this situation, we attach a triangle to the middle vertex of $P_{i,j}$ as in \Cref{fig:garbage-collector-gadget}.
We call this triangle a \emph{garbage collector gadget}.
The garbage collector gadget allows us to avoid making a white leaf as the other two vertices in the gadget would be black leaves.
\begin{figure}
    \centering
    \includegraphics{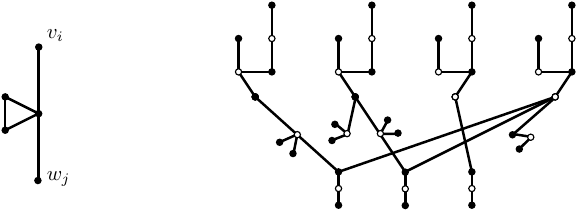}
    \caption{A garbage collector gadget and a forest that has no white leaves.}
    \label{fig:garbage-collector-gadget}
\end{figure}

To ensure the connectivity of variable gadgets, we add a connector gadget $H_i$ (illustrated in \Cref{fig:connector-gadget}) between two connection vertices $v_i$ and $v_{i+1}$ by identifying $a_i$ with $v_i$ and $b_i$ with $v_{i+1}$ for each $1 \le i < n$.
As shown in \Cref{fig:connector-gadget}, we can freely connect or disconnect $a_i$ and $b_i$ regardless of their colors, which enables us to connect isolated variable gadgets to the other part without introducing white leaves and undesired cycles.
We refer to the left three subgraphs that have a path between $a_i$ and $b_i$ as \emph{connector patterns} and to the other three subgraphs as \emph{disconnector patterns}.
In particular, a connector pattern (or disconnector pattern) is said to be \emph{consistent} (\emph{with two variable gadgets $G_i$ and $G_{i+1}$}) if $a_i$ and $b_i$ have the same color with $v_i$ and $v_{i+1}$, respectively.
\begin{figure}
    \centering
    \includegraphics[width=0.5\textwidth]{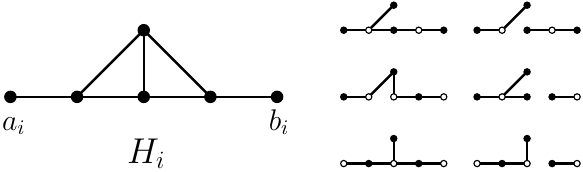}
    \caption{The left figure depicts the connector gadget $H_i$.
    }
    \label{fig:connector-gadget}
\end{figure}

The complete construction of $G_\varphi$ and a spanning even tree of $G_\varphi$ corresponding to some satisfying assignment of $\varphi$ are illustrated in \Cref{fig:complete}.
\begin{figure}
    \centering
    \includegraphics{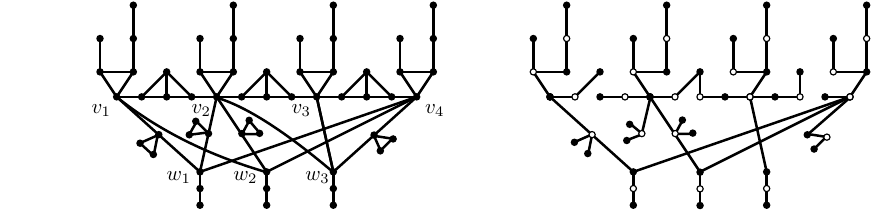}
    \caption{The figures illustrate the entire graph $G_\varphi$ for $\varphi = (\neg x_1 \lor \neg x_2 \lor x_4) \land (\neg x_1 \lor \neg x_2 \lor x_4) \land (x_2 \lor x_3 \lor \neg x_4)$ and a spanning even tree obtained from the assignment $(x_1, x_2, x_3, x_4) = (\texttt{false}, \texttt{false}, \texttt{true}, \texttt{true})$.}
    \label{fig:complete}
\end{figure}

\begin{lemma}\label{lem:hardness}
    There is a satisfying assignment of $\varphi$ if and only if $G_\varphi$ has a spanning even tree.
\end{lemma}
\begin{proof}
    From a satisfying assignment $\alpha$ of $\varphi$, we construct a subgraph $T$ of $G_\varphi$ as follows.
    For each variable~$x_i$, the subgraph $T$ contains $T_i$ if $\alpha(x_i) = \texttt{true}$ and contains $F_i$ otherwise.
    Moreover $T$ contains $P_j$ for all clauses $c_j$.
    For a variable~$x_i$ and a clause $c_j$ containing either $x_i$ or $\neg x_i$, 
    $T$ contains the edge of $P_{i, j}$ incident to $v_i$ as long as the addition does not make a cycle.
    Moreover, we add the other edge of $P_{i, j}$ to $T$ if $c_j$ contains $\neg x_i$ (that is, $P_{i, j}$ consists of two edges), $\alpha(x_i) = \texttt{false}$, and there is no path between these end vertices in $T$.
    This ensures that each clause gadget $P_j$ is connected to $v_i$ for some $i$ by a path $P_{i, j}$ as at least one literal of $c_j$ has to be $\texttt{true}$.
    For the internal vertex of $P_{i, j}$, we add two edges of the attached garbage collector gadget in such a way that white leaves never appear.
    Finally, for each $1 \le i < n$, we add one of the six subgraphs of connector gadget $H_i$ in such a way that if $a_i$ and $b_i$ belong to different connected components of $T$, we use the consistent connector pattern and otherwise we use the consistent disconnector pattern.
    Due to its construction, $T$ has no cycles and every leaf of $T$ is colored in black, yielding that $T$ is a spanning even tree of $G_\varphi$.

    Conversely, suppose that $G_\varphi$ has a spanning even tree $T$.
    Since each $v_i$ is a cut vertex in $G_\varphi$, the subgraph of $T$ induced by the vertices of variable gadget $G_i$ is a spanning tree $T$ of $G_i$.
    As the two vertices of degree~$1$ in $G_i$ must be leaves in this spanning tree, this spanning tree forms either $T_i$ or $F_i$.
    Now, we define a truth assignment $\alpha$ in such a way that $\alpha(x_i) = \texttt{true}$ if the spanning tree induced by the vertices of $G_i$ is $T_i$; otherwise $\alpha(x_i) = \texttt{false}$.
    This assignment is a satisfying assignment of $\varphi$ as each clause gadget $P_j$ is connected to $v_i$ for some $i$ due to the connectivity of $T$, which means that $\alpha(x_i) = \texttt{true}$ if $c_j$ contains $x_i$ (as a positive literal) and $\alpha(x_i) = \texttt{false}$ if $c_j$ contains $\neg x_i$.
\end{proof}

The reduction also shows that the problem is NP-complete even on planar graphs.
To see this, we give a reduction from \prb{Planar 3SAT}, which is known to be NP-complete~\cite{Lichtenstein82Planar}.
\prb{Planar 3SAT} is a special case of \prb{3SAT}, where the incidence bipartite graph of an input 3CNF formula is restricted to be planar.
Knuth and Raghunathan observed that \prb{Planar 3SAT} is still NP-complete even if the incident bipartite graph can be embedded into the plane so that all variable vertices are aligned along a straight line and no edge intersects this straight line~\cite{Lichtenstein82Planar,KnuthR92Problem} (see~\cref{fig:planar}~(a) for an illustration). Moreover, by a well-known technique of bounding occurrences of variables in \cite{DahlhausJPSY94,Tippenhauer2016}, the following restricted version of \prb{Planar 3SAT} is NP-complete: (1) each variable appears exactly three times, (2) each clause contains at most three literals, (3) the incidence graph is planar, and (4) all variable vertices are aligned along a straight line and any edge has no intersection with this straight line.
This enables us to construct a planar graph $G_\varphi$ of maximum degree 7 such that $\varphi$ is satisfiable if and only if $G_\varphi$ has a spanning even tree. Moreover, we replace each variable gadget as in \cref{fig:planar}~(b). It is easy to observe that the replaced graph has a spanning even tree if and only if the original graph has. Thus, we have the following.
\begin{figure}
    \centering
    \includegraphics[width=0.8\textwidth]{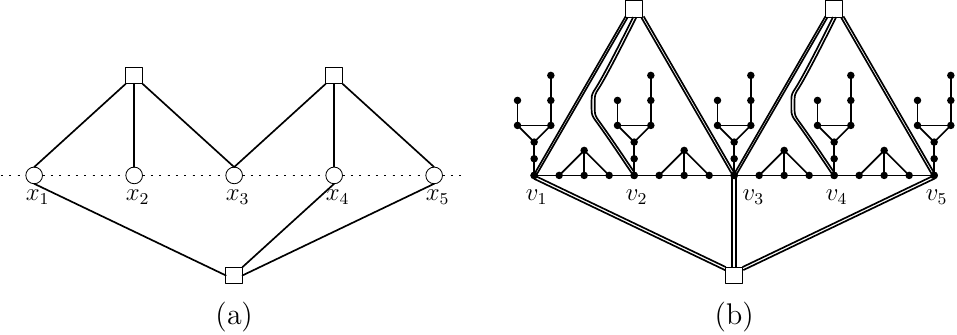}
    \caption{(a) A planar drawing of the incidence graph of an instance $\varphi$ of \prb{Planar 3SAT}.
    Circles and boxes indicate variable vertices and clause vertices, respectively.
    (b) The graph constructed from $\varphi$. For visibility, we omit some parts, which are not relevant in the construction.}
    \label{fig:planar}
\end{figure}

\begin{corollary}\label{cor:hardness-planar}
    The problem of deciding whether $G$ has a spanning even tree is NP-complete, even when the input graph $G$ is a planar graph of maximum degree 6.
\end{corollary}

\section{Polynomial-time algorithm for graph classes}

To circumvent the intractability shown in \cref{thm:hardness}, we consider several graph classes and devise polynomial-time algorithms for them.

\subsection{Cographs}

It is known that every bipartite regular graph has no spanning even trees~\cite{JacksonY24}.
This implies that the balanced complete bipartite graph $K_{t, t}$, which is also a cograph, has no spanning even tree.
The following theorem proves that excepting these graphs, every connected cograph has a spanning even tree.

\begin{theorem}\label{thm:cograph}
    Every connected cograph $G$ has a spanning even tree except for the balanced complete bipartite graph $K_{t, t}$ for all $t$.
\end{theorem}
\begin{proof}
    As $G$ is connected, there are two cographs $G_1 = (V_1, E_1)$ and $G_2 = (V_2, E_2)$ such that $G = G_1 \Join G_2$.
    
    Suppose first that $G$ is a complete bipartite graph.
    Without loss of generality, we assume that $|V_1| < |V_2|$.
    If $G$ is a star graph, we are done.
    Otherwise, we choose an arbitrary vertex $r \in V_2$.
    As $|V_1| < |V_2|$, there is a matching $M$ in $G - r$ saturating $V_1$.
    We define a tree that contains (1) all edges incident to $r$, (2) the matching $M$ saturating $V_1$, and (3) exactly one edge incident to $v$ for $v \in V_2 \setminus \{r\}$ that is not matched in $M$.
    Clearly, this is a spanning tree of $G$ and, as $|V_2| \ge 2$, all the leaves belong to $V_2$.
    Thus, $G$ has a spanning even tree.

    Suppose otherwise.
    If $|V_1| \neq |V_2|$, we can find a spanning even tree of $G$ since $G$ has $K_{|V_1|, |V_2|}$ as a spanning subgraph.
    Thus, we assume that $|V_1| = |V_2|$.
    We also assume that $V_1$ contains at least one edge~$e = \{u, v\}$.
    As $G - u$ is a connected cograph, either $G - u$ is a star graph with center $v$ or $G - u$ contains a spanning even tree $T$ such that all the leaves belong to $V_2$ and there is a vertex $r \in V_2$ that is not a leaf of $T$.
    If $G - u$ is a star graph, $G$ contains a spanning star with center $v$, and hence we are done.
    Otherwise, $T + e$ is indeed a spanning even tree of $G$.
    
    Therefore, $G$ has a spanning even tree under the assumption that $G \neq K_{t, t}$.
\end{proof}

The above proof immediately turns into a linear-time algorithm for finding a spanning even tree in a connected cograph.

\subsection{Cobipartite graphs}
Similarly to the case of cographs, every connected cobipartite graph has a spanning even tree, with a few exceptions.

\begin{theorem}\label{thm:cobiparite}
    Every connected cobipartite graph $G$ with at least two vertices has a spanning even tree except for the case that $G$ is isomorphic to a complete graph with two vertices, a path of four vertices, or a cycle of four vertices.
\end{theorem}
\begin{proof}
    Let $G = (V, E)$ be a connected cobipartite graph with at least two vertices.
    Let $V_1$ and $V_2$ be disjoint cliques in $G$ with $V_1 \cup V_2 = V$.
    Without loss of generality, we assume that $|V_2| \ge |V_1| \ge 1$.
    
    Suppose that $|V_1| = 1$.
    Then, $G$ has a spanning tree that is isomorphic to a star graph with at least two leaves, which is clearly even, under the assumption that $G \neq K_2$.
    
    Suppose that $|V_1| \ge 2$ and $|V_2| \ge 3$.
    Let $r \in V_2$ be an arbitrary vertex that is adjacent to a vertex $u$ in $V_1$.
    We construct a (roooted) spanning tree $T$ by adding $u$ and an arbitrary $v$ as children of $r$ and adding all the remaining vertices of $G$ as children of either $u$ or $v$.
    This can be done as all the vertices in $V_1$ are adjacent to $u$ and all the vertices in $V_2$ are adjacent to $v$.
    Since $|V_1| \ge 2$ and $|V_2| \ge 3$, both $u$ and $v$ have at least one child.
    As the depth of every leaf of $T$ is exactly $2$, $T$ is even.

    Now, the remaining case is $|V_1| = |V_2| = 2$.
    Since $G$ is isomorphic to neither a path of four vertices nor a cycle of four vertices, it must have a vertex of degree~$3$, which contains a star graph with three leaves as a subgraph.
\end{proof}

\subsection{Unit interval graphs}

In this subsection, we observe that every connected unit interval graph $G$ has a spanning even tree except for the case where $G$ is isomorphic to a path of odd length.

It is known that every connected unit interval graph has a Hamiltonian path~\cite{Bertossi83Finding}.
Such a Hamiltonian path can be found as follows.
Let $G$ be a connected unit interval graph.
We consider a unit interval representation of $G$ and the vertices of $G$ can be ordered from left to right with respect to this representation.
This vertex ordering $v_1, v_2, \ldots, v_n$ induces a Hamiltonian path of $G$, since two consecutive vertices $v_i$ and $v_{i+1}$ are not adjacent only if every interval corresponding to a vertex in $\{v_1, \ldots, v_i\}$ has no intersection with any interval corresponding to a vertex in $\{v_{i+1}, \ldots, v_n\}$, which implies that $G$ is disconnected.

We now turn to our problem.
If $G$ has an odd number of vertices, then $G$ admits a spanning even tree, as it has a Hamiltonian path of even length.
Suppose otherwise.
Since $G$ is not isomorphic to a path, $G$ has a vertex $v$ of degree at least three.
Moreover, as $G$ is claw-free~\cite{Roberts69}, $v$ is contained in a clique of size at least three.
Consider a unit interval representation of $G$.
Let $u_1, u_2, u_3$ be vertices of the clique that are consecutively ordered from left to right in the representation.
Since $u_1$ and $u_3$ are adjacent in $G$, $G - u_2$ is also connected.
Thus, $G - u_2$ has a Hamiltonian path of even length that traverses the unit interval representation induced by $G - u_2$ from left to right.
As $u_1$ and $u_3$ are adjacent in this Hamiltonian path, we can construct a spanning even tree of $G$, by attaching $u_2$ to $u_1$ or $u_3$.

\begin{theorem}\label{thm:unit-interval}
    Every connected unit interval graph has a spanning even tree, excepting that it is isomorphic to a path of odd length. 
\end{theorem}

By a linear-time algorithm for computing a unit interval representation~\cite{CorneilKNOS95}, we can decide if $G$ has a spanning even tree in linear time and construct it if exists.

\subsection{Split graphs}
Let $G = (I \cup K, E)$ be a connected split graph such that $I$ is a maximal independent set and $K$ is a clique of~$G$.
Note that, by the maximality of $I$, each vertex in $K$ has at least one neighbor in $I$.
Let $I^{(1)} \subseteq I$ be the set of degree-$1$ vertices in $G$ and let $K^{(1)} = N_G(I^{(1)})$.
As $I$ is an independent set of $G$, we have $K^{(1)} \subseteq K$.
Suppose first that $K^{(1)} = K$.
If $K$ contains at most one vertex, we can easily construct a spanning even tree of $G$, except for the case that $G$ is isomorphic to the complete graph of two vertices.
Thus, in the following, we assume that $K$ has at least two vertices.
The following lemma characterizes the existence of a spanning even tree of $G$ in this case, which readily turns into a linear-time algorithm for finding a spanning even tree.

\begin{lemma}\label{lem:split:K_1=K}
    Suppose that $K^{(1)} = K$ and $|K| \ge 2$.
    Let $\hat{G}$ be the bipartite graph obtained from $G$ by removing all edges between vertices in $K$.
    Then, $G$ has a spanning even tree if and only if $\hat{G}$ is connected.
\end{lemma}

\begin{proof}
    Let $T$ be a spanning even tree of $G$.
    Then, every vertex in $I^{(1)}$ is a leaf in $T$.
    As $K^{(1)} = K$, all vertices in $K$ are colored in white in the admissible coloring of any spanning even tree, and hence $K$ is an independent set of $T$.
    This implies that $T$ is also a spanning tree of $\hat{G}$.
    Thus, $\hat{G}$ is connected.

    Conversely, suppose that $\hat{G}$ is connected.
    Let $T$ be an arbitrary spanning tree of $\hat{G}$.
    As all vertices in $I^{(1)}$ are leaves in $T$ and $|K| \ge 2$, the vertices in $K = N_G(I^{(1)}) = N_{\hat{G}}(I^{(1)})$ are all internal vertices.
    As $K$ is an independent set of $\hat{G}$, the vertices in $T$ can be properly colored in such a way that the vertices in $K$ are colored in white and those in $I$ are colored in black.
    This proper coloring shows that $T$ is a spanning even tree of $\hat{G}$, and hence in $G$.
\end{proof}

Suppose next that $K^{(1)} \neq K$.
In this case, we show that $G$ always has a spanning even tree.
\begin{lemma}\label{lem:split:K_1!=K}
    Suppose that $K^{(1)} \neq K$.
    Then, $G$ has a spanning even tree.
\end{lemma}
\begin{proof}
    Let $r \in K\setminus K^{(1)}$ be arbitrary and let $I^{(2)} = \{v \in N_G(r) \cap I : d_G(v) = 2\}$ be the set of vertices in $I$ that have exactly two neighbors in $K$, one of which is $r$.
    Let $K' = K^{(1)} \cup N_G(I^{(2)}) \setminus \{r\}$ and let $I' = N_G(K') \cap I$.
    Since $K'$ contains $K^{(1)} = N_G(I^{(1)})$ and the neighbors of $I^{(2)}$ other than $r$, we have $I^{(1)} \cup I^{(2)} \subseteq I'$.
    Let $\tilde{K} = K \setminus (K' \cup \{r\})$ and let $\tilde{I} = I \setminus I'$.
    See~\Cref{fig:split} for an illustration.
    \begin{figure}
        \centering
        \includegraphics[width=0.6\textwidth]{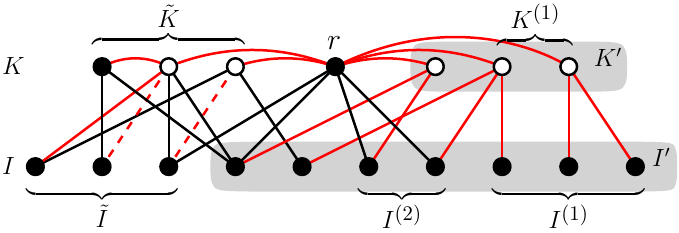}
        \caption{The figure illustrates the construction of a spanning even tree $T$ rooted at $r$. The edges in $T$ are represented by red lines and those in a maximal matching $M$ of $\tilde G$ by dashed red lines.}
        \label{fig:split}
    \end{figure}
    Clearly, every vertex in $\tilde{I}$ has no neighbor in $K'$.
    Moreover, as $G$ is connected and $\tilde I \cap I^{(1)} = \emptyset$, each vertex in $\tilde{I}$ has a neighbor in $\tilde K$.
    
    Now, we construct a (rooted) spanning tree $T$ of $G$ as follows.
    We take $r$ as its root and color it in black.
    For each vertex in $K'$, we add it as a child of $r$, which is colored in white.
    By the definition of $K'$, each vertex $u \in K'$ has a neighbor $v$ in $I'$ with $N_G(v) = \{u\}$ or $N_G(v) = \{u, r\}$.
    We add $v$ as a leaf child of $u$, and then all vertices in $K'$ become internal in $T$.
    Each remaining vertex in $I'$ is also added as a child of an arbitrary adjacent vertex in $K'$.
    All vertices in $I'$ are colored in black.
    If $\tilde K$ is empty, $\tilde I$ is also empty.
    Thus, the tree constructed so far is a spanning even tree of $G$ in this case.
    In the following, we assume that $\tilde K$ is nonempty.
    
    Suppose that $\tilde I$ is empty.
    If $K'$ contains at least one vertex $v$, which is already colored in white, then we add all the vertices of $\tilde K$ as black children of $v$, and hence we are done.
    Otherwise, i.e. $K' = \emptyset$, observe that $\tilde K$ contains at least two vertices, as otherwise the unique vertex in it belongs to $K^{(1)}$ or $N(I^{(2)})$, contradicting the fact that $\tilde K \cap K' = \emptyset$.
    Thus, we take an arbitrary vertex $v \in \tilde K$ and add it as a white child of $r$, and then all other vertices in $\tilde K$ are added as black children of $v$.
    As $\tilde K$ has at least two vertices, $v$ becomes an interval vertex of $T$, and hence we are done in this case.
    
    Suppose that $\tilde I$ is not empty.
    Let $\tilde G$ be the bipartite subgraph with vertex set $\tilde K \cup \tilde I$ and edge set $\{\{u, v\} \in E : u \in \tilde K, v \in \tilde I\}$.
    Let $M$ be an arbitrary maximal matching of $\tilde G$.
    Since each vertex in $\tilde I$ has no neighbor in $K'$ and at least two neighbors in $K$, it has a neighbor in $\tilde K$.
    This implies that $M$ is nonempty.
    For each vertex $u$ in $\tilde{K}$ that is matched in $M$, we add it as a white child of $r$ and add the other matched vertex $v \in \tilde{I}$ as a block child of $u$.
    Thus, all the matched vertices in $\tilde{K}$ are an internal vertex in $T$.
    For unmatched vertices in $\tilde K \cup \tilde I$, we add them as black children of a matched vertex in $\tilde K$.
    This can be done as $M$ is a (nonempty) maximal matching of $\tilde G$.
    The above construction confirms that $G$ always has a spanning even tree when $K^{(1)} \neq K$.
\end{proof}

The above two lemmas prove the following theorem.

\begin{theorem}\label{thm:split}
    There is a linear-time algorithm for deciding whether a given connected split graph has a spanning even tree.
    Moreover, the algorithm constructs a spanning even tree if the answer is affirmative.
\end{theorem}

\subsection{Block graphs}
Let $G = (V, E)$ be a connected block graph with at least three vertices.
Let $\mathcal B$ be the set of maximal cliques in $G$ and let $\mathcal T$ be the block-cut tree of $G$.
Before diving into the detail, we first observe the following simple fact.
\begin{observation}\label{obs:block:maximal-clique}
    Suppose that $G$ is a connected block graph that contains at least one edge.
    Let $T$ be an arbitrary spanning tree in $G$.
    Then, for every maximal clique $B$ in $G$, $T$ contains at least one edge between two vertices in $B$.
\end{observation}

We take an arbitrary node of $\mathcal T$ as its root and assume then that $\mathcal T$ is a rooted tree.
In the following, we color the vertices of $G$ in black or white and compute a spanning even tree of $G$ (if it exists), which is done by the following algorithm.
\begin{enumerate}\setlength{\parskip}{0cm}
    \item Color all degree-1 vertices of $G$ with black.
    \item Repeat the following exhaustively: If there is a maximal clique $B \in \mathcal B$ such that
    \begin{itemize}
        \item $B$ has exactly one uncolored vertex $v$ and
        \item the other vertices have the same color,
    \end{itemize}
    then color $v$ with the color opposite to the other vertices in $B$.
    \item Let $G'$ be the graph obtained from $G$ by deleting all edges whose end vertices have the same color.
    Note that every edge that has an uncolored end vertex remains in $G'$.
    If $G'$ has more than one components, answer ``No'' and halt.
    \item From the root of $\mathcal T$, repeat the following in the top-down manner: Choose a maximal clique $B \in \mathcal B$ that has at least one uncolored vertex $v$ and all the ancestral maximal cliques $B' \in \mathcal B$ with $B' \neq B$ has no uncolored vertices.
    Note that $B$ may not be a clique in $G'$.
    \begin{itemize}
        \item If $B$ has no white vertices, color $v$ in white and color all other (uncolored) vertices in black.
        \item Otherwise, color all uncolored vertices in black.
    \end{itemize}
    \item Let $G''$ be the graph obtained from $G'$ by deleting all edges whose end vertices have the same color. Compute an arbitrary spanning tree of $G''$ and output it. 
\end{enumerate}

Let $c'$ be the partial coloring of $G$ computed in Steps~1~and~2 and let $C' = \{v \in V : c'(v) \text{ is black or white}\}$.
The correctness of the above algorithm follows from the below lemmas.

\begin{lemma}\label{lem:block:cor1}
    For any spanning even tree $T$ of $G$, it holds that $c(v) = c'(v)$ for $v \in C'$, where $c$ is the admissible coloring for $T$.
\end{lemma}

\begin{proof}
    Let $T$ be a spanning even tree of $G$ and let $c$ be the admissible coloring of $T$.
    We define a total order~$\prec$ on $C'$ in which $u \prec v$ if and only if $u$ is colored before $v$ during the execution of the algorithm.
    Suppose that $C'$ contains a vertex $v$ with $c(v) \neq c'(v)$.
    Moreover, we assume that $v$ is minimal with respect to $\prec$, that is, $c(u) = c'(u)$ for $u \prec v$. 
    Since every vertex colored in Step~1 is a leaf in $T$, $c'(v)$ is determined in Step~2.
    This implies that $v$ belongs to a maximal clique $B$ such that for $u \in B \setminus \{v\}$, it holds that $u \prec v$ and $c'(u) \neq c'(v)$.
    Due to the minimality of $v$, we have $c(u) = c'(u)$ for $u \in B \setminus \{v\}$.
    Thus, all the vertices in $B$ have the same color under $c$.
    However, by~\Cref{obs:block:maximal-clique}, $T$ contains at least one edge between two vertices in $B$, which implies that its end vertices have different colors in $c$, a contradiction.
\end{proof}

\begin{corollary}\label{cor:block:cor1}
    For every spanning even tree $T$ of $G$, the graph $G'$ obtained in Step~3 contains $T$ as a subgraph.
\end{corollary}

It remains to show that the algorithm always outputs a spanning even tree under the assumption that $G'$ is connected.
This would prove the correctness of the above algorithm.

\begin{lemma}\label{lem:block:cor2}
    Let $G'$ be the graph obtained in Step~3 of the algorithm.
    Suppose that $G'$ is connected.
    Then, the algorithm outputs a spanning even tree.
\end{lemma}
\begin{proof}
    We first observe that every maximal clique in $B \in \mathcal B$ is not monochromatic under $c'$, as otherwise $G'$ is disconnected due to \Cref{obs:block:maximal-clique}.
    We also observe that every maximal clique in $B \in \mathcal B$ is not monochromatic under $c''$, where $c''$ is the coloring obtained after executing Step~4.
    This can be seen as follows.
    Since $B$ is not monochromatic under~$c'$, $B$ has at least two uncolored vertices under~$c'$, that is, $|B \cap C'| \ge 2$.
    However, if $B$ has an uncolored vertex when executing Step~4 for $B$, we can ensure that $B$ has vertices of both colors.
    Thus, such an uncolored vertex $v \in C'$ must be a cut vertex that is shared with some maximal clique $B' \in \mathcal B$, and it is colored in Step~4 for $B'$.
    As Step~4 is executed in the top-down manner, the cut vertex $v$ must be the parent of $B$ in $\mathcal T$.
    Since $|B \cap C'| \ge 2$, there is still at least one uncolored vertex in $B$, when executing Step~4 for $B$.

    Now we show that the algorithm always output a spanning even tree.
    Since every maximal clique in $\mathcal B$ is not monochromatic under $c''$, $G''$ is connected.
    Let $T$ be an arbitrary spanning tree of $G''$, and we claim that $T$ is even, that is, $T$ has no white leaves (under $c''$).
    Let $v$ be a white vertex.
    Suppose that $v$ belongs to a maximal clique $B \in \mathcal B$ of size exactly~$2$.
    In this case, $v$ must be a cut vertex of $G$, as otherwise $v$ is colored with black in Step~1.
    This implies that there is a maximal clique $B' \in \mathcal B$ that contains $v$.
    Since $B'$ is not monochromatic, $T$ must contain an edge between $v$ and a black vertex in $B'$, implying that the degree of $v$ is at least~$2$ in $T$.
    Suppose otherwise that all the maximal cliques that $v$ belongs to have at least three vertices.
    When the color of $v$ is determined in either Step~2 or Step~4, there is a maximal clique $B$ with $v \in B$ such that all vertices in $B \setminus \{v\}$ are colored with black.
    Then, $B$ induces a star with center $v$ and at least two leaves in $G''$, which implies that the degree of $v$ is at least~2 in~$T$.
    Therefore, $v$ is not a leaf in~$T$.
\end{proof}

Suppose that $G$ has a spanning even tree.
By~\Cref{cor:block:cor1}, $G'$ is connected, and hence, by~\Cref{lem:block:cor2}, the algorithm correctly outputs a spanning even tree in this case.
Suppose otherwise.
By~\Cref{lem:block:cor2}, $G'$ must be disconnected, and hence the algorithm correctly answers ``No'' in Step~3.
This completes the proof of the following theorem.

\begin{theorem}
    There is a linear-time algorithm for deciding whether an input block graph has a spanning even tree.
    Moreover, the algorithm computes a spanning even tree if the answer is affirmative.
\end{theorem}

\section{Conclusion}
In this paper, we discussed the complexity of the problem of finding a spanning even tree in a graph.
The problem is NP-hard even for planar graphs of maximum degree 7, while is polynomial-time solvable for bipartite graphs~\cite{EvenST:HoffmanW06}, for cographs, for cobipartite graphs, for unit interval graphs, for split graphs, and for block graphs.
It would be interesting to seek more algorithmic or complexity results of this problem on several graph classes, such as chordal graphs, interval graphs, and comparability graphs.
Furthermore, it would be worth considering the complexity on graphs of maximum degree 3 as it is in P if the input graph is 3-regular. 

\section*{Acknowledgments}
The authors are grateful to Kiyoshi Yoshimoto for bringing this topic to their attention and providing the manuscript of \cite{JacksonY24}.

\begin{sloppypar}
\printbibliography
\end{sloppypar}

\end{document}